\documentclass[12pt]{article}

\usepackage[latin1]{inputenc}
\usepackage{amssymb}
\usepackage{amsmath}
\usepackage{amsthm}
\usepackage{latexsym}
\usepackage{subfigure}
\usepackage{graphicx}
\usepackage{bm}
\usepackage{overpic}
\usepackage[normalem]{ulem}
\usepackage{url}
\usepackage{soul}
\usepackage{exscale}
\usepackage{amsfonts}
\usepackage[usenames,dvipsnames]{color} 
\usepackage{hyperref}
\usepackage{verbatim}
\usepackage[usenames,dvipsnames]{xcolor}

\newtheorem{definition}{Definition}[section]
\newtheorem{lemma}{Lemma}[section]
\newtheorem{remark}{Remark}[section]
\newtheorem{example}{Example}[section]
\newtheorem{theorem}{Theorem}[section]
\newtheorem{cor}{Corollary}[section]

\numberwithin{equation}{section}

\newcommand{\Ex}{\mathrm{e}}
\newcommand{\E}{\mathbb{E}}
\newcommand{\R}{\mathbb{R}}
\newcommand{\Dd}{\mathcal{D}}

\newcommand{\BS}{\rm BS}
\newcommand{\p}{\partial}

\newcommand{\beas}{\begin{eqnarray*}}
\newcommand{\eeas}{\end{eqnarray*}}
\newcommand{\bea}{\begin{eqnarray}}
\newcommand{\eea}{\end{eqnarray}}
\newcommand{\ben}{\begin{enumerate}}
\newcommand{\een}{\end{enumerate}}

\newcommand{\bi}{\begin{itemize}}
\newcommand{\ei}{\end{itemize}}
\newcommand{\beq}{\begin{equation}}
\newcommand{\eeq}{\end{equation}}
\newcommand{\bv}{\begin{verbatim}}
\newcommand{\ev}{\end{verbatim}}

\newcommand{\sgn}{\mathrm{sign}}
\definecolor{shade}{gray}{.99}

\definecolor{shade}{gray}{.99}

\begin{document}

\title{\bf Arbitrage-free SVI volatility surfaces}

\author{Jim Gatheral\footnote{Department of Mathematics, Baruch College, CUNY. {\tt  jim.gatheral@baruch.cuny.edu}}{\setcounter{footnote}{1}} , Antoine Jacquier\footnote{Department of Mathematics, Imperial College, London. {\tt  ajacquie@imperial.ac.uk}}{\setcounter{footnote}{2}} }


\maketitle\thispagestyle{empty}


\begin{abstract}

In this article, we show how to calibrate the widely-used SVI parameterization of the implied volatility smile in such a way as to guarantee the absence of static arbitrage.  
In particular, we exhibit a large class of arbitrage-free SVI volatility surfaces with a simple closed-form representation.  
We demonstrate the high quality of typical SVI fits with a numerical example using recent SPX options data.
\end{abstract}

%
%
%
%

\section{Introduction}

The {\em stochastic volatility inspired} or {\em SVI} parameterization of the implied volatility smile was originally devised at Merrill Lynch in 1999 and subsequently publicly disseminated in \cite{Gatheral:2004}.  This parameterization has two key properties that have led to its popularity with practitioners:
\begin{itemize}
\item{For a fixed time to expiry $t$, the implied Black-Scholes variance~$\sigma_{\BS}^2(k,t)$ is linear in the log-strike $k$ as $|k| \to \infty$ consistent with Roger Lee's moment formula~\cite{RogerLee}.}
\item{It is relatively easy to fit listed option prices whilst ensuring no calendar spread arbitrage.}
\end{itemize}

The consistency of the SVI parameterization with arbitrage bounds for extreme strikes has also led to its use as an extrapolation formula~\cite{HypHyp}.
Moreover, as shown in~\cite{GatheralJacquier}, the SVI parameterization is not arbitrary in the sense that the large-maturity limit of the Heston implied volatility smile is exactly SVI.  
However it is well-known that SVI smiles may be arbitrageable.
Previous work has shown how to calibrate SVI to given implied volatility data (for example \cite{Zeliade}).  
Other recent work \cite{CarrWu} has been concerned with showing how to parameterize the volatility surface in such a way as to preclude dynamic arbitrage.  
There has been some work on arbitrage-free interpolation of implied volatilities or equivalently of option prices \cite{AndreasenHuge}, \cite{Fengler}, \cite{GlaserHeider}, \cite{Kahale}.  
Prior work has not successfully attempted to eliminate static arbitrage and indeed, efforts to find simple closed-form arbitrage-free parameterizations of the implied volatility surface are still widely considered to be futile.

In this article, we exhibit a large class of SVI volatility surfaces with a simple closed-form representation, for which absence of static arbitrage is guaranteed.
Absence of static arbitrage---as defined by Cox and Hobson~\cite{CoxHobson}---corresponds 
to the existence of a non-negative martingale on a filtered probability space such that European call option prices can be written as the expectation, under the risk-neutral measure, of their final payoffs.
This definition also implies (see~\cite{Fengler}) that the corresponding total variance must be an increasing function of the maturity (absence of calendar spread arbitrage).
Using some mathematics from the Renaissance, we show how to eliminate any calendar spread arbitrage resulting from a given set of SVI parameters.
We also present a set of necessary conditions for the corresponding density to be non-negative (absence of butterfly arbitrage), 
which corresponds---from the definition of static arbitrage---to call prices being decreasing and convex functions of the strike.
We go on to use the existence of such arbitrage-free surfaces to devise a new algorithm for eliminating butterfly arbitrage should it occur.  
With both types of arbitrage eliminated, we achieve a volatility surface that typically calibrates well to given implied volatility data and is guaranteed free of static arbitrage.

In Section~\ref{sec:calenderSpreadArbitrage}, we present a necessary and sufficient condition for the absence of calendar spread arbitrage.
In Section~\ref{sec:butterflySpreadArbitrage}, we present a necessary and sufficient condition for the absence of butterfly arbitrage, or negative densities.
In Section~\ref{sec:parameterizations}, we present various equivalent forms of the SVI parameterization.
In Section~\ref{sec:ssvi}, we exhibit a large and useful class of SVI volatility surfaces that are guaranteed to be free of static arbitrage.
In Section~\ref{sec:calibration}, we show how to calibrate SVI to observed option prices, avoiding both butterfly and calendar spread arbitrages.  We further show how to interpolate and extrapolate in such a way as to guarantee the absence of static arbitrage.
Finally, in Section \ref{sec:conclusion}, we summarize and conclude.

\textbf{Notations.}
In the foregoing, we consider a stock price process $\left(S_t\right)_{t\geq 0}$
with natural filtration $\left(\mathcal{F}_t\right)_{t\geq 0}$, and we define the forward price process $\left(F_t\right)_{t\geq 0}$
by $F_t:=\mathbb{E}\left(S_t|\mathcal{F}_0\right)$.
For any $k\in\mathbb{R}$ and $t>0$,
$C_{\BS}(k,\sigma^2 t)$ denotes the Black-Scholes price of a European Call option on $S$ with strike $F_t \Ex^{k}$,
maturity $t$ and volatility $\sigma>0$.
We shall denote the Black-Scholes implied volatility by $\sigma_{\BS}(k,t)$, and define the total implied variance by
$$
w(k,t)=\sigma_{\BS}^2(k,t)t.
$$
The implied variance $v$ shall be equivalently defined as $v(k,t)=\sigma_{\BS}^2(k,t)=w(k,t)/t$.
We shall refer to the two-dimensional map $(k,t)\mapsto w(k,t)$ as the volatility surface,
and for any fixed maturity $t>0$, the function $k\mapsto w(k,t)$ will represent a slice.
We propose below three different---yet equivalent---slice parameterizations of the total implied variance,
and give the exact correspondence between them.
For a given maturity slice, we shall use the notation $w(k;\chi)$ where $\chi$ represents a set of parameters, and drop the $t$-dependence.

\section{Characterisation of static arbitrage}\label{sec:arb}
In this section we provide model-independent definitions of (static) arbitrage and some preliminary results.
We define static arbitrage for a given volatility surface in the following way, 
which is equivalent to the definition of static arbitrage for call options 
recalled in the introduction (see also~\cite{Roper}).
\begin{definition}\label{def:NoArbVol}
A volatility surface is free of static arbitrage if and only if the following conditions are satisfied:
\begin{itemize}
\item[(i)] it is free of calendar spread arbitrage;
\item[(ii)] each time slice is free of butterfly arbitrage.
\end{itemize}
\end{definition}
In particular, absence of butterfly arbitrage ensures the existence of a (non-negative) probability density, and absence of calendar spread arbitrage implies monotonicity of option prices with respect to the maturity.
The following two subsections analyse in details each of these two types of arbitrage, 
in a model-independent way.

\subsection{Calendar spread arbitrage}\label{sec:calenderSpreadArbitrage}

Calendar spread arbitrage is usually expressed as the monotonicity of European call option prices with respect to the maturity (see for example~\cite{CarrMadan} or~\cite{Cousot}).  Since our main focus here is on the implied volatility, we translate this definition into a property of the implied volatility.
Indeed, assuming proportional dividends, we establish a necessary and sufficient condition for an implied volatility parameterization to be free of calendar spread arbitrage.
This can also be found in~\cite{Fengler} and~\cite{Gatheral:2004} and we outline its proof for completeness.

\begin{lemma}\label{lem:noCalendarArb}
If dividends are proportional to the stock price, the volatility surface~$w$ is free of calendar spread arbitrage if and only if
$$
\partial _t w(k,t) \geq 0,
\quad\text{for all }k\in\mathbb{R}\text{ and }t>0.
$$
\end{lemma}

\begin{proof}
Let $\left(X_t\right)_{t\geq 0}$ be a martingale, $L\geq 0$ and $0\leq t_1<t_2$.
Then the inequality
$$
\E\left[(X_{t_2}-L)^+\right]\geq \E\left[(X_{t_1}-L)^+\right]
$$
is standard.
For any $i=1,2$, let $C_i$ be options with strikes $K_i$ and expirations $t_i$.
Suppose that the two options have the same moneyness, i.e.
$$
\frac{K_1}{F_{t_1}}=\frac{K_2}{F_{t_2}}=:\Ex^k
$$
Then, if dividends are proportional, {the process $(X_t)_{t\geq 0}$ defined by }
$X_t := S_t/F_{t}$ for all $t\geq 0$ is a martingale and
$$
\frac{C_2}{K_2}=\Ex^{-k}\E\left[\left(X_{t_2}-\Ex^{k}\right)^+\right] \geq
\Ex^{-k}\E\left[\left(X_{t_1}-\Ex^{k}\right)^+\right]=\frac{C_1}{K_1}
$$
So, if dividends are proportional, keeping the moneyness constant, option prices are non-decreasing in time to expiration.
The Black-Scholes formula for the non-discounted value of an option may be expressed in the form
{$C_{\BS}(k, w(k,t))$  with $C_{\BS}$ strictly increasing in its second argument.
It follows that for fixed $k$, the function $w(k,\cdot)$ must be non-decreasing.}
\end{proof}

\noindent Lemma \ref{lem:noCalendarArb} motivates the following definition.

\begin{definition}\label{def:CalendarArb}
{A volatility surface $w$ is free of calendar spread arbitrage if
$$
\partial _t w(k,t) \geq 0,
\quad\text{for all }k\in\mathbb{R}\text{ and }t>0.
$$
}
\end{definition}

\subsection{Butterfly arbitrage}\label{sec:butterflySpreadArbitrage}

In Section~\ref{sec:calenderSpreadArbitrage}, we provided conditions under which a volatility surface could be guaranteed to be free of calendar spread arbitrage.
We now consider a different type of arbitrage, namely butterfly arbitrage (Definition~\ref{def:ButterflyArb}).  Absence of this arbitrage corresponds to the existence of a risk-neutral martingale measure and the classical definition of no static arbitrage, 
as developed in~\cite{Follmer} or~\cite{CoxHobson}.
In this section, we consider only one slice of the implied volatility surface, i.e. the map $k\mapsto w(k,t)$
for a given fixed maturity $t>0$.
For clarity we therefore drop---in this section only---the $t$-dependence of the smile and use the notation $w(k)$ instead.
Unless otherwise stated, we shall always assume that the map $k\mapsto w(k,t)$ is at least of class $\mathcal{C}^2(\mathbb{R})$ for all $t\geq 0$.

\begin{definition}\label{def:ButterflyArb}
A slice is said to be free of butterfly arbitrage if the corresponding density is non-negative.
\end{definition}
Recall the Black-Scholes formula for a European call option price:
$$
C_{\BS}(k,w(k))=S\left(\mathcal{N}(d_+(k))-\Ex^{k}\mathcal{N}(d_-(k))\right),
\quad\text{for all }k\in\mathbb{R},
$$
where $\mathcal{N}$ is the Gaussian cdf and 
$d_{\pm}(k):=-k/\sqrt{w(k)}\pm\sqrt{w(k)}/2$.
Let us define the function $g:\mathbb{R}\to\mathbb{R}$ by
\begin{equation}\label{eq:gBut}
g(k):=\left(1-\frac{k w'(k)}{2 w(k)}\right)^2-\frac{w'(k)^2}{4}\left(\frac{1}{w(k)}+\frac{1}{4} \right)
+\frac{w''(k)}{2}.
\end{equation}
This function will be the main ingredient in the determination of butterfly arbitrage as stated in the following lemma.
\begin{lemma}\label{lem:NoButterflyArb}
A slice is free of butterfly arbitrage if and only if $g(k)\geq 0$ for all $k\in\mathbb{R}$
and $\lim\limits_{k\to+\infty}d_+(k)=-\infty$.
\end{lemma}

\begin{proof}
It is well known~\cite{Breeden} that the probability density function $p$ may be computed from the call price function $C$ as
$$
p(k) = \left.\frac{\partial^2 C(k)}{\partial K^2}\right|_{K=F_t\Ex^k}
 = \left.\frac{\partial^2 C_{\BS}(k,w(k))}{\partial K^2}\right|_{K=F_t\Ex^k},
\quad\text{for any }k\in\mathbb{R}.
$$
Explicit differentiation of the Black-Scholes formula then gives for any $k\in\mathbb{R}$,
$$
p(k)=\frac{g(k)} {\sqrt{2 \pi w(k)}}\exp\left(-\frac{d_-(k)^2}{2}\right).
$$
We have so far implicitly assumed that the function $p$ is a well-defined density, 
and in particular that it integrates to one.
This may not always be the case though, and one needs to impose asymptotic boundary conditions.
In particular, call prices must converge to 0 as $k$ tends to infinity, 
which is equivalent to having $\lim_{k\to+\infty}d_+(k)=-\infty$.
We refer the reader to~\cite{RT10} for a proof of this equivalence.
\end{proof}

\section{SVI formulations}\label{sec:parameterizations}
We first recall here the original SVI formulation proposed in~\cite{Gatheral:2004}, and then present some alternative (but equivalent) ones.
We emphasize in particular that even though the original (``raw'') formulation is very tractable and has become popular 
with practitioners, it is difficult---seemingly impossible---to find precise conditions on the parameters to prevent arbitrage.

\subsection{The raw SVI parameterization}
For a given parameter set $\chi_R=\{a,b,\rho, m , \sigma\}$, the {\em raw SVI parameterization} of total implied variance reads:
\beq
w(k;\chi_R)=a+b\,\left\{\rho\,(k-m)+\sqrt{(k-m)^2+\sigma^2}\right\},
\label{eq:SVIraw}
\eeq
where $a\in\mathbb{R}$, $b\geq 0$, $|\rho|<1$, $m\in\mathbb{R}$, $\sigma>0$,
and the obvious condition $a+b\,\sigma\,\sqrt{1-\rho^2}\geq 0$, 
which ensures that $w(k;\chi_R)\geq 0$ for all $k\in\mathbb{R}$.
This condition indeed ensures that the minimum of the function $w(\cdot;\chi_R)$ is non-negative.
Note further that the function $k\mapsto w(k;\chi_R)$ is (strictly) convex on the whole real line.
It follows immediately that changes in the parameters have the following effects:
\begin{itemize}
\item Increasing $a$ increases the general level of variance, a vertical translation of the smile;
\item Increasing $b$ increases the slopes of both the put and call wings, tightening the smile;
\item Increasing $\rho$ decreases (increases) the slope of the left(right) wing, a counter-clockwise rotation of the smile;
\item Increasing $m$ translates the smile to the right;
\item Increasing $\sigma$ reduces the at-the-money (ATM) curvature of the smile.
\end{itemize}
{We exclude the trivial cases $\rho=1$ and $\rho=-1$, where the volatility smile is respectively strictly increasing and decreasing.
We also exclude the case $\sigma=0$ which corresponds to a linear smile.}

\subsection{The natural SVI parameterization}

For a given parameter set $\chi_N=\{\Delta, \mu,\rho,\omega,\zeta \}$, the {\em natural SVI parameterization} of total implied variance reads:
\begin{equation}\label{eq:SVInatural}
w(k;\chi _N) =
\Delta+\frac{\omega}{2}\left\{1+\zeta\rho \left(k - \mu\right) +
\sqrt{\left(\zeta(k-\mu)+\rho \right)^2+\left(1-\rho^2\right)}\right\},
\end{equation}
where $\omega\geq 0$, $\Delta\in\mathbb{R}$, $\mu\in\mathbb{R}$, $|\rho|<1$ and $\zeta>0$.
It is straightforward to derive the following correspondence between the raw
and natural SVI parameters:
\begin{lemma}
We have the following mapping of parameters between the raw and the natural SVI:
\begin{equation}\label{eq:NaturalToRaw}
\left(a, b, \rho, m, \sigma\right)
 =
\left(\Delta+\frac{\omega}{2}\left(1-\rho^2\right), \frac{\omega\zeta}{2}, \rho, \mu-\frac{\rho}{\zeta}, \frac{\sqrt{1-\rho^2}}{\zeta}\right),
\end{equation}
and its inverse transformation, between the natural and the raw SVI:
\begin{equation}\label{eq:RawToNatural}
\left(\Delta, \mu,\rho,\omega,\zeta \right)
=
\left(
a-\frac{\omega}{2}\left(1-\rho^2\right),
m+\frac{\rho\sigma}{\sqrt{1-\rho^2}},
\rho,
\frac{2b\sigma}{\sqrt{1-\rho^2}},
\frac{\sqrt{1-\rho^2}}{\sigma}\right).
\end{equation}
\end{lemma}

\subsection{The SVI Jump-Wings (SVI-JW) parameterization}

Neither the raw SVI nor the natural SVI parameterizations are intuitive to traders in the sense that a trader cannot be expected to carry around the typical value of these parameters in his head.
Moreover, there is no reason to expect these parameters to be particularly stable.
The {\em SVI-Jump-Wings (SVI-JW) parameterization}  of the implied variance $v$ (rather than the implied total variance $w$) was inspired by a similar parameterization attributed to Tim Klassen, then at Goldman Sachs.
For a given time to expiry $t>0$ and a parameter set $\chi_J=\{v_t,\psi_t,p_t,c_t,\widetilde v_t\}$ the SVI-JW parameters are defined from the raw SVI parameters as follows:
\begin{equation}\label{eq:rawToJW}
\left.
\begin{array}{rll}
v_t & = \displaystyle \frac{a+  b \,\left\{-\rho\,m+\sqrt{m^2+\sigma^2} \right\}}{t},\\
\psi_t & = \displaystyle \frac{1}{ \sqrt{w_t}}\,\frac{ b}{2}\,\left( - \frac{m}{{\sqrt{m^2 + { \sigma}^2}}} + \rho \right),\\
p_t & = \displaystyle \frac{1}{\sqrt{w_t}}b\left(1-\rho\right),\\
c_t & = \displaystyle \frac{1}{ \sqrt{w_t}}b\left(1+\rho\right),\\
\widetilde v_t & = \displaystyle \frac{1}{t}\left(a+ b\,\sigma\,\sqrt{1-\rho^2}\right),
\end{array}
\right.
\end{equation}
with $w_t:= v_tt$.
Note that this parameterization has an explicit dependence on the time to expiration $t$, and hence can be viewed as generalizing the raw (expiration-independent) SVI parameterization.
The SVI-JW parameters have the following interpretations:
\begin{itemize}
\item{$v_t$ gives {the} ATM variance;}
\item{$\psi_t$ gives {the} ATM skew;}
\item{$p_t$ gives the slope of the left (put) wing;}
\item{$c_t$ gives the slope of the right (call) wing;}
\item{$\widetilde v_t$ is the minimum implied variance.}
\end{itemize}

If smiles scaled perfectly as $1/\sqrt{w_t}$ (as is approximately the case empirically), these parameters would be constant, independent of the slice $t$.  
This makes it easy to extrapolate the SVI surface to expirations beyond the longest expiration in the data set.
Also note that by definition, for any $t>0$ we have
\[
\psi_t=\left.\frac{\partial \sigma_{\BS}(k,t)}{\partial k}\right|_{k=0}
\]
The choice of volatility skew as the skew measure rather than variance skew for example, reflects the empirical observation that volatility is roughly lognormally distributed.
Specifically, following the lines of~\cite[Chapter 7]{jimbook}, assume that the instantaneous variance process satisfies the SDE
\[
 dv_t = \alpha (v_t)\,dt   + \eta \sqrt{v_t}\,\beta(v_t)\, dZ_t,
\quad \text{for all }t\geq 0
\]
where $\eta>0$, $(Z_t)_{t\geq 0}$ is a standard Brownian motion
and $\alpha$ and $\beta$ two functions on $\mathbb{R}_+$ ensuring the existence of a unique strong solution to the SDE
(see for instance~\cite{KarSh} for exact conditions),
then the ATM variance skew
$$
\left.\lim_{t\to 0}\frac{\partial \sigma_{\BS}(k,t)^2} {\partial k}\right|_{k=0}
$$
exists and is proportional to~$\beta(v)$. 
If the variance process is lognormal so that~$\beta(v)$ behaves like~$\sqrt{v}$, the limit of the at-the-money {\em volatility} skew as time to expiry tends to zero is constant and independent of the volatility level.
This consistency of the SVI-JW parameterization with empirical volatility dynamics thus leads in practice to greater parameter stability over time.
The following lemma provides the inverse representation of~\eqref{eq:rawToJW}.
\begin{lemma}\label{lem:SVIJWtoSVI}
Assume that $m\ne 0$.
For any $t>0$, define the ($t$-dependent) quantities:
$$
\beta := \rho-\frac{2\psi_t\sqrt{w_t}}{b}
\quad\text{and}\quad
\alpha := \sgn(\beta)\sqrt{\frac{1}{\beta^2}-1}.
$$
where we have further assumed that $\beta \in [-1,1]$\footnote{The condition $\beta \in [-1,1]$ is equivalent to $-p_t \leq 2\psi_t \leq c_t$, i.e. to the convexity of the smile.}.
Then, the raw SVI and SVI-JW parameters are related as follows:
\begin{eqnarray*}
b&=&\frac{\sqrt{w_t}}{2}\left(c_t+p_t\right),\\
\rho&=&1-\frac{p_t\,\sqrt{w_t}}{b},\\
a&=&\widetilde v_{t}t-b\sigma\sqrt{1-\rho^2},\\
m &=&
\frac{\left(v_t-\widetilde v_{t}\right)t}{b\left\{-\rho+\sgn(\alpha)\sqrt{1+\alpha^2}-\alpha\sqrt{1-\rho^2}\right\}},
\\
\sigma&=& \alpha\, m.
\end{eqnarray*}
If $m=0$, then the formulae above for $b$, $\rho$ and $a$ still hold, but
$\sigma = \left(v_t t-a\right)/b$.
\end{lemma}

\begin{proof}
The expressions for $b$, $\rho$ and $a$ follow directly from~\eqref{eq:rawToJW}.
Assume that $m\ne 0$ and let $\beta:=\rho-2\psi_t \sqrt{w_t}/b$ and $\alpha:=\sigma/m\in\mathbb{R}$.
Then the expressions in~\eqref{eq:rawToJW} give
$$\beta = \frac{\sgn\left(\alpha\right)}{\sqrt{1+\alpha^2}},$$
which implies that
$$\alpha = \sgn(\beta)\sqrt{\frac{1}{\beta^2}-1}.$$
Using~\eqref{eq:rawToJW}, we also have
$$
\frac{\left(v_t-\widetilde v_{t}\right)\,t}{b}
=m\left\{-\rho+\sgn(\alpha)\,\sqrt{1+\alpha^2}-\alpha\,\sqrt{1-\rho^2}\right\},
$$
from which we deduce $m$ in terms of $\alpha$, and
the expression of $\sigma$ is recovered from the equality $\sigma=\alpha m$.
The expression for $\sigma$ in the case $m=0$ is straightforward from~\eqref{eq:rawToJW}.
\end{proof}

\subsection{Arbitrage and absence thereof in SVI parameterizations}
Given a volatility surface, it is natural to wonder whether it is free of arbitrage.
Since we can easily switch from any of the SVI formulations to either of the other two using
Lemma~\ref{eq:NaturalToRaw} and Lemma~\ref{lem:SVIJWtoSVI}, we shall state the following results only for the raw SVI parameterization~\eqref{eq:SVIraw}.
Referring to~\eqref{eq:SVIraw} as a volatility surface is a slight abuse of language 
since~\eqref{eq:SVIraw} is really an expiry-independent slice parameterization.
A volatility surface is thus understood as a (discrete) collection of slices, with a different set of parameters for each expiry.
Checking calendar arbitrage in the sense of Lemma~\ref{lem:noCalendarArb} is then equivalent to checking 
for calendar arbitrage for any pair of expiries $t_1$ and $t_2$.
The following lemma establishes a sufficient condition for the absence of calendar spread arbitrage.

\begin{lemma}
The raw SVI surface~\eqref{eq:SVIraw} is free of calendar spread arbitrage if a certain quartic polynomial (given in~\eqref{eq:Quartic} below) has no real root.
\end{lemma}

\begin{proof}
By definition, there is no calendar arbitrage if for any two dates $t_1\ne t_2$,
the corresponding slices $w\left(\cdot,t_1\right)$ and $w\left(\cdot,t_2\right)$ do not intersect.
Let these two slices be characterised by the sets of parameters
$\chi_1:=\left\{a_1,b_1,\sigma_1,\rho_1,m_1\right\}$ and
$\chi_2:=\left\{a_2,b_2,\sigma_2,\rho_2,m_2\right\}$,
and assume for convenience that $0<t_1<t_2$.
We therefore need to determine the (real) roots of the equation
$w\left(k,t_1\right)=w\left(k,t_2\right)$.
The latter is equivalent to
\begin{equation}\label{eq:SVIIntersect}
a_1 + b_1\left\{\rho_1 \left(k - m_1\right)  + \sqrt{\left(k - m_1\right)^2 +\sigma^2_1}\right\}
= a_2 + b_2\left\{\rho_2 \left(k -m_2\right) + \sqrt {\left(k - m_2\right)^2+\sigma^2_2 }\right\}.
\end{equation}
Leaving $\sqrt{\left(k - m_1\right)^2 +\sigma^2_1}$ on one side, squaring the equality and rearranging it leads to
$$
2b_2\left(\alpha+\beta k\right)\sqrt{\left(k-m_2\right)^2+\sigma^2_2}
 = b_1^2\left\{\left(k-m_1\right)^2+\sigma^2_1 \right\}-b_2^2\left\{\left(k-m_2\right)^2 +\sigma^2_2\right\}
 -\left(\alpha+\beta k\right)^2,
$$
where
$\alpha:=a_2-a_1+b_1\rho_1 m_1-b_2\rho_2 m_2$
and
$\beta:=b_2\rho_2-b_1\rho_1$.
Squaring the last equation above gives a quartic polynomial equation of the form
\begin{equation}\label{eq:Quartic}
\alpha_4\,k^4+\alpha_3\,k^3+\alpha_2\,k^2+\alpha_1\,k+\alpha_0=0,
\end{equation}
where each of the coefficients lengthy yet explicit expressions\footnote{Explicit expressions for these coefficients can be found in the R-code posted on \url{http://faculty.baruch.cuny.edu/jgatheral}.} in terms of
the parameters $\left\{a_1,b_1,\rho_1,\sigma_1,m_1\right\}$ and
$\left\{a_2,b_2,\rho_2,\sigma_2,m_2\right\}$.
If this quartic polynomial has no real root, then the slices do not intersect and the lemma follows.
Roots of a quartic polynomial are known in closed-form thanks to Ferrari and Cardano~\cite{Cardano}.
Thus there exist closed-form expressions in terms of $\chi_1$ and $\chi_2$
for the possible intersection points of the two SVI slices.
\end{proof}

\begin{remark}
If the quartic polynomial~\eqref{eq:Quartic} has one or more real roots,
we need to check whether the latter are indeed solutions of the original problem~\eqref{eq:SVIIntersect},
or spurious solutions arising from the two squaring operations.
The absence of real roots of the quartic polynomial is clearly a sufficient---but not necessary---condition.
\end{remark}
\begin{remark}
By a careful study of the minima and the shapes of the two slices $w(\cdot,t_1)$ and $w(\cdot,t_2)$,
it is possible to determine a set of conditions on the parameters ensuring no calendar spread arbitrage.
However these conditions involve tedious combinations of the parameters and will hence not match the computational simplicity of the lemma.
\end{remark}

For a given slice, we now wish to determine conditions on the parameters of the raw SVI formulation~\eqref{eq:SVIraw}
such that butterfly arbitrage is excluded.
By Lemma~\ref{eq:gBut}, this is equivalent to showing 
(i) that the function~$g$ defined in~\eqref{eq:gBut} is always positive and 
(ii) that call prices converge to zero as the strike tends to infinity.
Sadly, the highly non-linear behaviour of~$g$ makes it seemingly impossible to find general conditions on the parameters that would eliminate butterfly arbitrage.
We provide below an example where butterfly arbitrage is violated.
Notwithstanding our inability to find general conditions on the parameters that would preclude arbitrage, in Section~\ref{sec:sviSurfaces}, we will introduce a new sub-class of SVI volatility surface for which the absence of butterfly arbitrage is guaranteed for all expiries.

\begin{example}\label{ex:AxelVogt}(From Axel Vogt on wilmott.com)
Consider the raw SVI parameters:
\begin{equation}\label{eq:VogtParams}
\left(a, b, m, \rho, \sigma\right) =
\left(-0.0410,0.1331,0.3586,0.3060,0.4153\right),
\end{equation}
with $t=1$.
These parameters give rise to the total variance smile~$w$ and the function~$g$
(defined in~\eqref{eq:gBut}) on Figure~\ref{fig:VogtVol}, where
the negative density is clearly visible.
\begin{figure}[htb!]
\begin{center}
\subfigure{\includegraphics[scale=0.27]{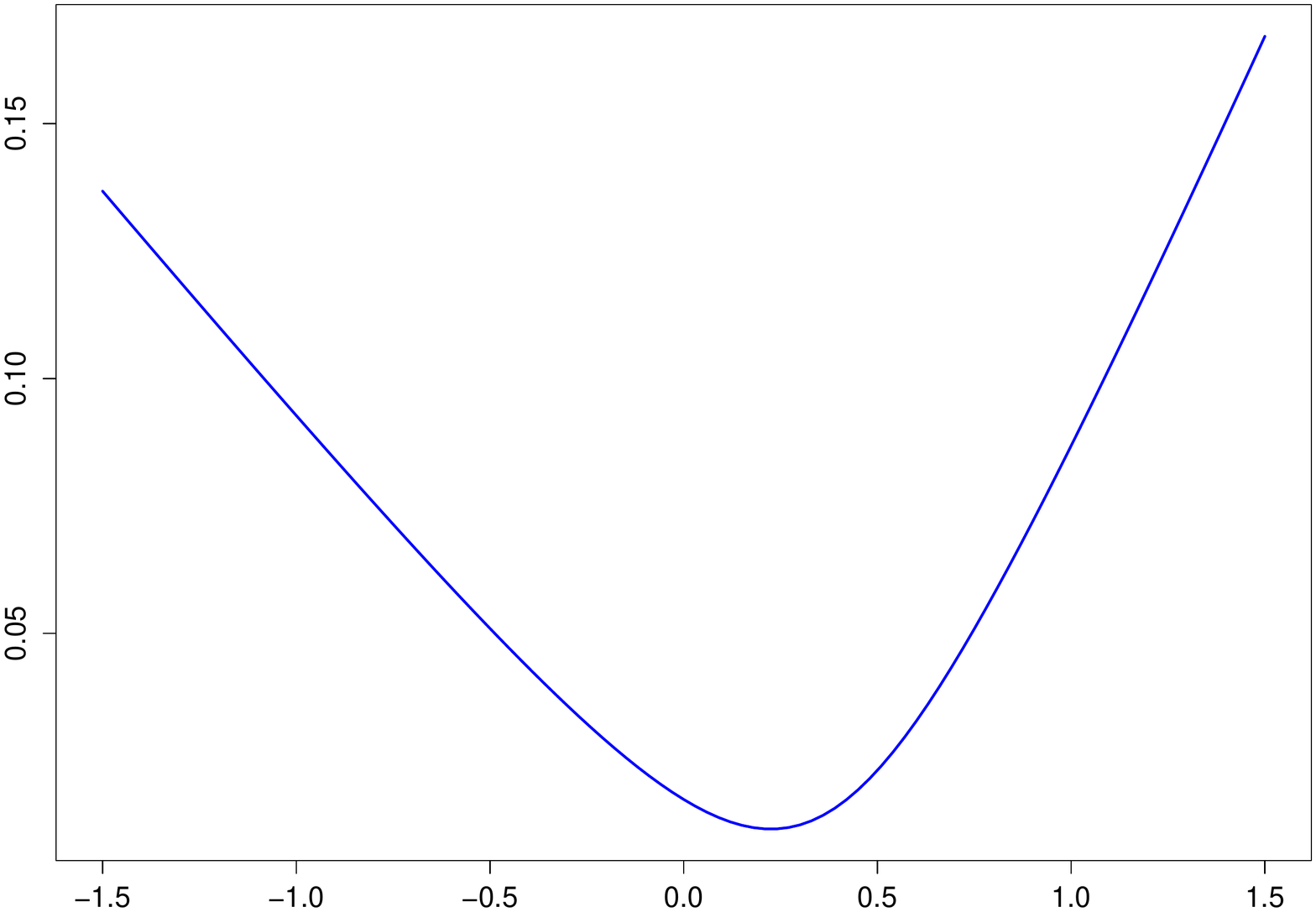}}
\subfigure{\includegraphics[scale=0.27]{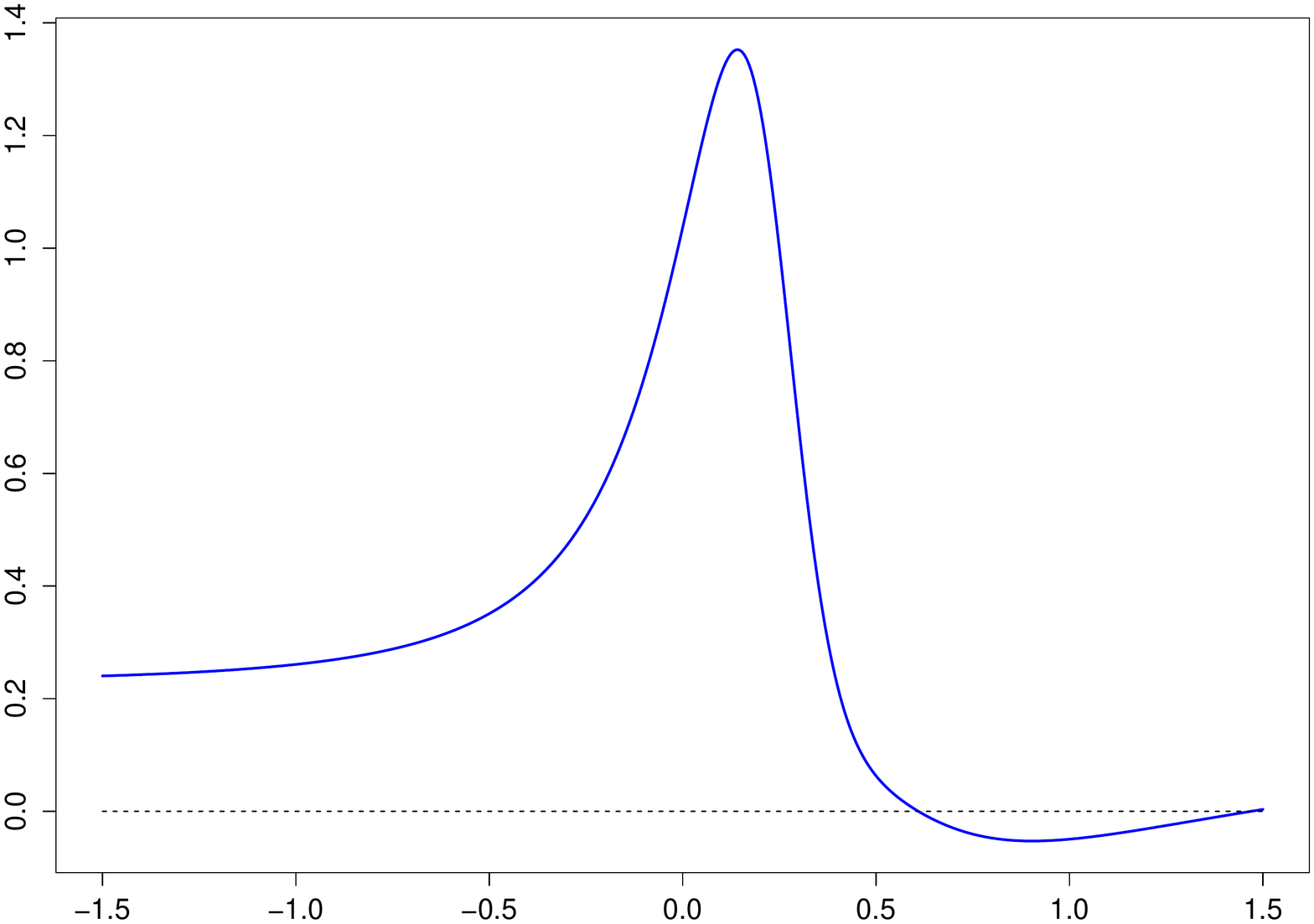}}
\caption{Plots of the total variance smile~$w$ (left) and the function~$g$ defined in~\eqref{eq:gBut} (right), using the parameters~\eqref{eq:VogtParams}.}
\label{fig:VogtVol}
\end{center}
\end{figure}
\end{example}

\section{Surface SVI: A surface free of static arbitrage}\label{sec:sviSurfaces}\label{sec:ssvi}

We now introduce a class of SVI volatility surfaces---which we shall call SSVI (for `Surface SVI')---as 
an extension of the natural parameterization~\eqref{eq:SVInatural}.
For any maturity~$t\geq 0$, define the at-the-money (ATM) implied total variance~$\theta_t:=\sigma_{\BS}^2(0,t)t$.
We shall assume that the function~$\theta$ is at least of class $\mathcal{C}^1$ on $\R_+^*$.
An ATM option with zero time to expiry has no value so $\theta_0:=\lim_{t\to 0}\theta_t=0$.
\begin{definition}
Let $\varphi$ be a smooth function from $\R_+^*$ to $\R_+^*$  such that the limit 
$\lim_{t\to 0}\theta_t\varphi(\theta_t)$ exists in~$\mathbb{R}$.
We refer to as {\em SSVI} the surface defined by
\beq\label{eq:ssvi}
w(k,\theta_t)
= \frac{\theta_t}{2}\left\{1+\rho\varphi(\theta_t) k + \sqrt{\left(\varphi(\theta_t){k}+\rho\right)^2
+(1-\rho^2 })\right\}.
\eeq
\end{definition}
From Section~\ref{sec:parameterizations}, SSVI corresponds to 
the natural SVI volatility surface parameterization~\eqref{eq:SVInatural} with
$\chi_N = \left\{0,0,\rho,\theta_t,\varphi(\theta_t)\right\}$.
Note that this representation amounts to considering the volatility surface in terms of ATM variance time, instead of standard calendar time, similar in spirit to the stochastic subordination of~\cite{Clark}.

\begin{remark}
In the parameterization~\eqref{eq:ssvi}, the ATM variance curve~$\theta_t$ may be viewed as a (vector) parameter of the volatility surface.  
Moreover, this parameter is directly observable given market prices for a finite set of expiries, and can be considered well-known to traders even for expiries which are not explicitly quoted.
The explicit reference to~$\theta_t$ also emphasizes the importance of studies such as~\cite{DeMarco} of the ATM variance structure in classical models which may shed some light on how to impose dynamics on SSVI.
\end{remark}

The ATM implied total variance is $\theta_t=\sigma^2_{\BS}(0,t)\,t$ and the ATM volatility skew is given by
\begin{equation}\label{eq:ATMVolSkew}
\left.\partial_k \sigma_{\BS}(k,t)\right|_{k=0}
=\left.\frac{1}{2\sqrt{\theta_t t}}\p_k w(k,\theta_t)\right|_{k=0}
=\frac{\rho\sqrt{\theta_t}}{2\sqrt{t}}\varphi(\theta_t).
\end{equation}
Furthermore the smile is symmetric around at-the-money if and only if $\rho=0$.
This is consistent with~\cite[Theorem 3.4]{CarrLee} which states that in a standard stochastic volatility model,
the smile is symmetric if and only if the correlation between the stock price and its instantaneous volatility is null.
Since $\theta_0=0$, we have at time $t=0$:
\beq
w(k,\theta_0)=\frac{1}{2}\,\phi_0\,\left(\rho k+|k|\right), \quad\text{for any }k\in\mathbb{R},
\label{eq:timeZeroSmile}
\eeq
where $\phi_0:=\lim_{\theta\to 0}\theta\varphi(\theta)$.
$\phi_0=0$ is characteristic of stochastic volatility models as in Example~\ref{ex:HestonExample};
$\phi_0>0$ as in Example~\ref{ex:PowerLaw} gives a V-shaped time zero smile which is characteristic of models with jumps and in particular, characteristic of empirically observed volatility surfaces.
For notational convenience, we shall always assume that $\lim_{t\nearrow\infty}\theta_t=\infty$.
As proved in~\cite{RT10}, this is equivalent (assuming no interest rate) to the stock price 
(assumed to be a non-negative martingale) to converging to zero as $t$ tends to infinity.
Although this holds in many popular models (Black-Scholes, Heston, exponential L\'evy), 
this is not always true, see~\cite{Hobson} for counter-examples.
If $\lim_{t\nearrow\infty}\theta_t$ is finite, all our results remain valid, but only on the support of 
the function $t\mapsto \theta_t$.

The following theorem gives precise necessary and sufficient conditions to ensure that the SSVI volatility surface~\eqref{eq:ssvi}
is free of calendar spread arbitrage (Lemma~\ref{lem:noCalendarArb})
and also matches the term structure of ATM volatility and the term structure of the ATM volatility skew.
\begin{theorem}\label{thm:arbitrarySkew}
The SSVI surface~\eqref{eq:ssvi} is free of calendar spread arbitrage if and only if 
\begin{enumerate}
\item $\partial_t \theta_t \geq 0$\label{item:increasingTotVar}, for all $t\geq 0$;
\item $0\leq \partial_\theta (\theta \varphi(\theta))  \leq 
\frac{1}{\rho^2}\left(1+\sqrt{1-\rho^2}\right)\varphi(\theta)$, for all $\theta>0$,
\end{enumerate}
where the upper bound is infinite when $\rho=0$.
\end{theorem}

In particular, this theorem implies that the SSVI surface~\eqref{eq:ssvi} 
is free of calendar spread arbitrage if the skew in total variance terms is
monotonically increasing in trading time and the skew in implied variance terms is monotonically decreasing in trading time.
In practice, any reasonable skew term structure that a trader defines has these properties.

\begin{proof}
Since the definition of calendar spread arbitrage does not depend on the log-moneyness $k$,
there is no loss of generality in assuming $k$ fixed.
First note that
$\partial_t w(k,\theta_t)=\partial_\theta w(k,\theta_t)\partial_t \theta_t$
so the SSVI volatility surface~\eqref{eq:ssvi} is free of calendar spread arbitrage if
$\partial_\theta w(k,\theta) \geq 0$ for all $\theta>0$.

Consider first the case $|\rho|<1$.
To proceed, we compute, for any $\theta>0$,
$$
2\partial_\theta w (k,\theta)=\psi_0(x,\rho)
+\gamma(\theta) \psi_1(x,\rho),
$$
with $x:= k\varphi(\theta)$, $\gamma(\theta) := \partial_\theta(\theta\varphi(\theta))/\varphi(\theta)$,
$$
\psi_0(x,\rho) := 1+\frac{1+\rho x }{\sqrt{x^2+2\rho x +1}}
\quad\text{and}\quad
\psi_1(x,\rho):= x \left\{\frac{x+\rho }{\sqrt{x^2+2\rho x +1}}+\rho \right\}.
$$
For any $|\rho|<1$, $\psi_0(x,\rho)$ is strictly positive for all $x\in\R$.
Now define the set
$$
\Dd_\rho =
\left\{
\begin{array}{ll}
(-\infty,0)\cup (-2\rho,\infty),
\quad & \text{if }\rho<0,\\
(-\infty,-2\rho)\cup (0,\infty),
\quad & \text{if }\rho>0,\\
\R\setminus\{0\},
\quad & \text{if }\rho=0.
\end{array}
\right.
$$
Then 
$\psi_1(\cdot,\rho) >0 \text{ if } x \in \Dd_\rho$
and 
$\psi_1(\cdot,\rho) <0 \text{ if } x \in \R \setminus \left( \Dd_\rho\cup\{0,-2\,\rho\} \right)$.  
It follows that
\beq
\partial_\theta w(k,\theta)\geq 0 \text{ if and only if }
\left\{
\begin{array}{ll}
\displaystyle \gamma(\theta)\geq -\frac{\psi_0(x,\rho)}{\psi_1(x,\rho)}, 
 & \text{for }x\in\Dd_\rho,\\
\displaystyle \gamma(\theta)\leq -\frac{\psi_0(x,\rho)}{\psi_1(x,\rho)}, 
 & \text{for }x\in\R\setminus\left(\Dd_\rho\cup\{0,-2\,\rho\}\right),\\
\end{array}
\right.
\label{eq:iffConditions}
\eeq
When $x\in\{0,-2\rho\}$, then $\psi_1(x,\rho)=0$ and so $\partial_\theta w(k,\theta)\geq 0$.  
The inequalities \eqref{eq:iffConditions} thus give necessary and sufficient conditions 
for absence of calendar spread arbitrage for any given $x \in \R$.
To determine the tightest possible bounds on $\gamma(\theta)$, we compute
$$
\sup_{x\in\Dd_\rho}-\frac{\psi_0(x,\rho)}{\psi_1(x,\rho)} = 0
\quad\text{and}\quad
\inf_{x\in\R\setminus\left(\Dd_\rho\cup\{0,-2\rho\}\right)}-\frac{\psi_0(x,\rho)}{\psi_1(x,\rho)} = 
\frac{1+\sqrt{1-\rho^2}}{\rho^2}.
$$
The supremum in the first equality is never attained (the function increases to zero from below as $|x|$ tends to infinity).  However the infimum in the second equality is attained at $x=-\rho\notin\mathcal{D}_\rho$.  It follows that
$$
\partial_\theta w(k,\theta)\geq 0 \text{ if and only if }
0 \leq  \gamma(\theta)\leq \frac{1+\sqrt{1-\rho^2}}{\rho^2}.
$$
Note that when $\rho=0$, the infimum above is taken over an empty set, and there is hence no upper bound.

When $\rho=1$, for any $(k,\theta)\in\R\times (0,\infty)$, we have
$$
\p_\theta w(k,\theta) = \left( 1+\frac{1+x}{\sqrt{(1+x)^2}} \right)\,\left(1 + \gamma(\theta)\,x\right)
= \left\{\begin{array}{ll}
2\,(1 + \gamma(\theta)\,x) &\text{ if } x \geq -1,\\
0 & \text{ otherwise}.
\end{array}
\right.
$$
Obviously, $\p_\theta w(k,\theta) \geq 0$ if $x \geq 0$.  
For $x>-1$,
clearly $\p_\theta w(k,\theta) \geq 0$ if and only if $\gamma(\theta)\in [0,1]$.
Similarly, with $\rho=-1$, we have
$$
\p_\theta w(k,\theta) = \left( 1+\frac{1-x}{\sqrt{(1-x)^2}} \right)\,\left(1 - \gamma(\theta)\,x\right)
 = \left\{\begin{array}{ll}
2\,(1 - \gamma(\theta)\,x) &\text{ if } x \leq 1,\\
0 & \text{ otherwise}.
\end{array}
\right.
$$
Again $\p_\theta w(k,\theta) \geq 0$ if $x \leq 0$, and for $x\leq 1$,
$\p_\theta w(k,\theta) \geq 0$
if and only if 
$\gamma(\theta)\in [0,1]$.
\end{proof}

The following lemma is a straightforward consequence of~\eqref{eq:NaturalToRaw} and~\eqref{eq:rawToJW}.
\begin{lemma}\label{lem:ssviToJW}
The SVI-JW parameters associated with the SSVI surface~\eqref{eq:ssvi} are
\beas
v_t&=&\theta_t/t,\\
\psi_t&=&\frac{1}{2}\,\rho\,\sqrt{\theta_t}\,\varphi(\theta_t),\\
p_t&=&\frac12\,\sqrt{\theta_t}\,\varphi(\theta_t)\,(1-\rho),\\
c_t&=&\frac 12\,\sqrt{\theta_t}\,\varphi(\theta_t)\,(1+\rho),\\
\widetilde v_t&=&\frac{\theta_t}{t}\,(1-\rho^2).
\eeas
\end{lemma}

We now give several examples of SSVI implied volatility surfaces~\eqref{eq:ssvi}.
\begin{example}{\bf  A Heston-like parameterization}\label{ex:HestonExample}\\
\noindent Consider the function $\varphi$ defined by
$$
\varphi(\theta)\equiv\frac{1}{\lambda\theta}\left\{1-\frac{1-\Ex^{-\lambda\theta}}{\lambda\theta}\right\},
$$
with $\lambda>0$.
Then for all $\theta>0$, we immediately obtain
$$
\p_\theta\left(\theta\varphi(\theta)\right)
=\frac{\Ex^{-\lambda\theta}\left(\Ex^{ \lambda\theta}-1-\lambda\theta\right)}{\lambda^2\theta^2}>0
\quad\text{and}\quad
\frac{\p_\theta\left(\theta\varphi(\theta)\right)}{\varphi(\theta)}
 = \frac{1-(1+\lambda\theta)\Ex^{-\lambda\theta}}{\Ex^{-\lambda\theta}+\lambda\theta-1}.
$$
For any $\lambda>0$, the map $\theta\mapsto \p_\theta\left(\theta\varphi(\theta)\right) / \varphi(\theta)$
is strictly decreasing on $(0,\infty)$ with limit as $\theta$ tends to zero equal to one.
Since the quantity $(1+\sqrt{1-\rho^2})/\rho^2$ is greater than one for any $\rho\in [-1,1]$, 
the conditions of Theorem~\ref{thm:arbitrarySkew} are satisfied.
This function is consistent with the implied variance skew in the Heston model as shown in~\cite[Equation 3.19]{jimbook}.
\end{example}

\begin{example}{\bf  Power-law parameterization}\label{ex:PowerLaw}\\
\noindent Consider
$\varphi(\theta)=\eta\theta^{-\gamma}$
with $\eta>0$ and $0<\gamma<1$.
Then 
$\p_\theta\left(\theta\varphi(\theta)\right) / \varphi(\theta) = 1-\gamma \in(0,1)$
holds for all $\theta>0$, and hence the conditions of Theorem~\ref{thm:arbitrarySkew} are satisfied.
In particular if $\gamma=1/2$ then Lemma~\ref{lem:ssviToJW} implies that the SVI-JW parameters $\psi_t$, $p_t$, and $c_t$ associated with the SSVI volatility surface~\eqref{eq:ssvi} are constant and independent of the time to expiration~$t$.
Furthermore, Equation~\ref{eq:ATMVolSkew} implies that the ATM volatility skew is given by
$$
\left.\p_k \sigma_{\BS}(k,t)\right|_{k=0}=\frac{\rho\,\eta}{2\sqrt{t}}.
$$
\end{example}

The following theorem provides sufficient conditions for a SSVI surface~\eqref{eq:ssvi} to be free of butterfly arbitrage.

\begin{theorem}\label{thm:Butterfly}
The SSVI volatility surface~\eqref{eq:ssvi} is free of butterfly arbitrage if the following conditions are satisfied
for all $\theta>0$:
\begin{enumerate}
\item $\theta\varphi(\theta)\left(1+|\rho|\right)< 4$;\label{item:RLasymptote}
\item $\theta\varphi(\theta)^2\left(1+|\rho|\right)\leq 4$.\label{item:phi2condition}
\end{enumerate}
\end{theorem}

\begin{proof}
For ease of notation, we suppress the explicit dependence of $\theta$ and $\varphi$ on~$t$.
By symmetry, it is enough to prove the theorem for $0\leq \rho<1$.
We shall therefore assume so, and we define $z:=\varphi k$.
The function $g$ defined in~\eqref{eq:gBut} reads
$$
g(z)=\frac{f(z)}{64 \left(z^2+2z\rho +1\right)^{3/2}},
$$
where
$$
f(z):=a-b\varphi^2\theta-\frac{c}{16}\varphi^2\theta^2,
$$
and where~$a$, $b$ and~$c$ depend on~$z$.
In the following, we frequently use the inequality
$$
z^2+2 z \rho +1=(z+\rho)^2+1-\rho^2 \ge 0.
$$
Computing the coefficient of $\varphi^2\theta^2$ in~$f(z)$ explicitly gives
\begin{align*}
c & = \sqrt{z^2+2 z \rho +1} \left\{\left(1+\rho ^2\right)\left(z+\rho\right)^2+2\rho (z+\rho)\sqrt{z^2+2 z \rho+1}+\left(1-\rho^2\right) \rho ^2\right\}\\
  & \geq \sqrt{z^2+2 z \rho +1} \left\{\left(1+\rho^2\right)\left(z+\rho\right)^2+2 \rho  (z+\rho )^2 +\left(1-\rho^2\right)\rho^2\right\}\\
  & = \sqrt{z^2+2 z \rho +1}\left\{\left(1+\rho\right)^2 \left(z+\rho\right)^2+\left(1-\rho^2\right) \rho ^2\right\}
  \geq 0.
\end{align*}
Thus if
$$
0\leq \theta\varphi \leq \frac 4 {1+\rho}
\quad \text{and} \quad
0\leq \theta\varphi^2\leq \frac 4 {1+\rho},
$$
we have
$$
f(z) \geq \left\{
\begin{array}{ll}
\displaystyle a - \frac{4\,b}{1+\rho}-\frac{c}{(1+\rho)^2}=:f_1(z), & \text{ if } b\geq 0, \\
\displaystyle  a -\frac{c}{(1+\rho)^2}=:f_2(z),  & \text{ if }    b<0.
\end{array}
\right.
$$

It is then straightforward to verify that
\begin{align*}
\frac{2f_1(z)}{(1+\rho)^2} &
 = \sqrt{z^2+2 z \rho +1} \left\{z^2 \rho -z (1-\rho ) \rho +2 (1+\rho) \left(1-\rho ^2\right)+\rho \right\}\\
 & +\rho \left(z+\rho\right)^2+3 \rho  \left(1-\rho ^2\right)+2 \left(1-\rho ^2\right)-z \rho\left(z^2+2z\rho+1\right),
\end{align*}
which is clearly positive for $z<0$.
To see that $f_1(z)$ is also positive when $z>0$, we rewrite it as
\beas
&&\frac{2\,f_1(z)}{(1+\rho)^2}\\
&=&\left\{\sqrt{z^2+2 z \rho +1}-(z+\rho) \right\} \left\{\rho  \left(z-\frac{1-\rho }{2}\right)^2+2 (1+\rho) \left(1-\rho^2\right)+\rho \left(1-\frac{(1-\rho )^2}{4}\right)\right\}\\
   && + (1+\rho) \left\{z \left(2-\rho^2\right)+2\left(1+\rho\right)\left(1-\rho^2\right)+\rho\right\}.
\eeas
Consider now the function $f_2(z)$.  It is straightforward to verify that
$$
f_2(z)=-\frac{2 z^3 \rho }{(1+\rho)^2}+\left(z^2+2 z \rho +1\right)^{3/2}
+2\left(z^2+2z\rho+1\right)+\sqrt{z^2+2z\rho+1}
$$
which is positive by inspection if $z<0$.  To see that $f_2(z)$ is also positive when $z>0$, we rewrite it as
\begin{align*}
f_2(z) & = z^3\,\frac{ 1+\rho ^2}{(1+\rho)^2}+3 z^2 \rho +2 \left(z^2+2 z \rho +1\right)\\
& +\left(z^2+2 z \rho +1\right)\left\{\sqrt{z^2+2 z \rho +1}-(z+\rho) \right\}\\
& +\sqrt{z^2+2 z \rho +1}+2 z \rho ^2+z+\rho.
\end{align*}
Thus $f(z) \geq 0$ in all cases.
From Lemma~\ref{lem:NoButterflyArb}, we are left to prove that $\lim_{k\to\infty}d_+(k)=~-~\infty$.
A straightforward computation shows that this is satisfied as soon as Condition 1 in Theorem~\ref{thm:Butterfly} holds.
\end{proof}

\begin{remark}
A SSVI volatility surface~\eqref{eq:ssvi} is free of butterfly arbitrage if
$$
\sqrt{v_t\,t}\max\left(p_t,c_t\right) < 2,
\quad\text{and}\quad
(p_t+c_t)\max\left(p_t,c_t\right)\leq 2,
$$
hold for all $t>0$.  The proof follows from Lemma \ref{lem:ssviToJW} by re-expressing Conditions \ref{item:RLasymptote} and \ref{item:phi2condition} of Theorem \ref{thm:Butterfly} in terms of SVI-JW parameters.
\end{remark}

The following lemma shows that Theorem \ref{thm:Butterfly} is almost if-and-only-if.

\begin{lemma}
The SSVI volatility surface~\eqref{eq:ssvi} is free of butterfly arbitrage only if
$$
\theta\varphi(\theta)\left(1+|\rho|\right) \leq 4,
\quad\text{for all }\theta>0.
$$
Moreover if $\theta\varphi(\theta)\left(1+|\rho|\right)=4$, the SSVI surface is free of butterfly arbitrage only if
$$
\theta\varphi(\theta)^2\left(1+|\rho|\right)\leq 4.
$$
Thus Condition~\ref{item:RLasymptote} of Theorem~\ref{thm:Butterfly} is necessary and Condition~\ref{item:phi2condition} is tight.
\end{lemma}
\begin{proof}
Considering the SSVI surface~\eqref{eq:ssvi} and the function $g$ defined in~\eqref{eq:gBut},
we have
$$
g(k) =
\left\{
\begin{array}{ll}
\displaystyle \frac{16-\theta^2\varphi(\theta)^2\left(1+\rho\right)^2}{64}+\frac{4-\theta\varphi(\theta)^2 \left(1+\rho\right)}{8\varphi(\theta) k}+\mathcal{O}\left(\frac{1}{k^2}\right),
\quad\text{as }k\to+\infty,\\
\displaystyle \frac{16-\theta^2\varphi(\theta)^2\left(1-\rho\right)^2}{64}-\frac{4-\theta\varphi(\theta)^2\left(1-\rho\right)}
{8\varphi(\theta)k}+\mathcal{O}\left(\frac{1}{k^2}\right),\quad\text{as }k\to-\infty.
\end{array}
\right.
$$
The result follows by inspection.
\end{proof}

\begin{remark}\label{rem:AsymptoticSlopes}
The asymptotic behavior of SSVI~\eqref{eq:ssvi} as $|k|$ tends to infinity is
$$
w(k,\theta_t) =\frac{\left(1\pm \rho\right)\theta_t}{2}\varphi(\theta_t)\left|k\right|+\mathcal{O}(1),
\quad\text{for any }t>0.
$$
We thus observe that the condition $\theta\varphi(\theta)\left(1+|\rho|\right)\leq 4$ of Theorem~\ref{thm:Butterfly}
 corresponds to the upper bound of $2$ on the asymptotic slope established by Lee~\cite{RogerLee} and so again, Condition~\ref{item:RLasymptote} of Theorem~\ref{thm:Butterfly} is necessary.
\end{remark}

The following corollary follows directly from Theorems~\ref{thm:arbitrarySkew} and~\ref{thm:Butterfly}.

\begin{cor}\label{cor:NoStaticArb}
The SSVI surface~\eqref{eq:ssvi} is free of static arbitrage if the following conditions are satisfied:
\begin{enumerate}
\item $\partial_t \theta_t \geq 0$, for all $t>0$
\item $0\leq \partial_\theta (\theta \varphi(\theta))  \leq 
\frac{1}{\rho^2}\left(1+\sqrt{1-\rho^2}\right)\varphi(\theta)$, for all $\theta>0$;
\item $\theta\varphi(\theta)\left(1+|\rho|\right)< 4$, for all $\theta>0$;
\label{item:4-2}
\item $\theta\varphi(\theta)^2 \left(1+|\rho|\right)\leq 4$, for all $\theta>0$.\label{item:5-2}
\end{enumerate}
\end{cor}

\begin{remark}\label{rmk:arbFreePowerLaw}
Consider the function $\varphi(\theta)= \eta\theta^{-\gamma}$ with $\eta>0$ from Example~\ref{ex:PowerLaw}, then Condition~\ref{item:phi2condition} imposes $\gamma\in (0,1)$.
From Condition \ref{item:4-2}, such surfaces can be free of static arbitrage only up to some maximum expiry.
Take for instance the simple case $\theta_t:=\sigma^2 t$ for some $\sigma>0$.
Then the map $\psi:t\mapsto \theta_t\varphi(\theta_t)\left(1+|\rho|\right)- 4$ is clearly strictly increasing
with $\psi(0)=-4$ and $\lim_{t\to\infty}\psi(t)=\infty$.
Therefore there exists $t^*_0>0$ such that $\psi(t)\leq 0$ for $t\leq t^*_0$.
The map $\psi_2:t\mapsto \theta_t\varphi(\theta_t)^2\left(1+|\rho|\right)- 4$ is
\begin{itemize}
\item strictly increasing if $\gamma\in (0,1/2)$ with $\psi_2(0)=-4$ and $\lim\limits_{t\to\infty}\psi(t)=+\infty$;
there exists $t^*_1>0$ such that $\psi_2(t)\leq 0$ for $t\leq t^*_1$.
\item strictly decreasing if $\gamma\in (1/2, 1)$ with $\lim\limits_{t\to 0}\psi_2(0)=+\infty$ and $\lim\limits_{t\to\infty}\psi(t)=-4$;
there exists $t^*_1>0$ such that $\psi_2(t)\leq 0$ for $t\geq t^*_1$.
\item constant if $\alpha=1/2$ with $\psi_2\equiv-4$.
\end{itemize}
When $\gamma\in (0,1/2)$, the surface is guaranteed to be free of static arbitrage only for $t\leq t^*_0\wedge t^*_1$.
For $\gamma\in (1/2,1)$, this remains true only for $t\in (0,t^*_0)\cap(t^*_1,\infty)$ (which may be empty).
When $\gamma=1/2$, static arbitrage cannot occur for $t\leq t^*_0$.
However, the behavior for large~$\theta$ can be easily modified so as to ensure that the entire surface is free of static arbitrage.
For example, the choice
\begin{equation}\label{eq:arbFreePhi}
\varphi(\theta)= \frac{\eta}{\theta^\gamma\,(1+\theta)^{1-\gamma}}
\end{equation}
gives a surface that is completely free of static arbitrage provided that $\eta\left(1+|\rho|\right)\leq 2$.
\end{remark}

\begin{remark}
In the Heston-like parameterization of Example~\ref{ex:HestonExample}, note that
$$
\lim_{\theta\to+\infty}\theta\varphi(\theta)\left(1+|\rho|\right) = \frac{1+|\rho|}{\lambda}.
$$
Therefore Condition~\ref{item:4-2} of Corollary~\ref{cor:NoStaticArb} imposes
$\lambda\geq \left(1+|\rho|\right)/4$.
\end{remark}

The following model-independent theorem provides a way to expand the class of volatility surfaces that are guaranteed to be free of static arbitrage by adding a suitable time-dependent function.

\begin{theorem}\label{thm:AddingAlpha}
Let~$(k,t)\mapsto w(k,t)$ be a SSVI volatility surface~\eqref{eq:ssvi} satisfying the conditions of Corollary~\ref{cor:NoStaticArb}
(in particular free of static arbitrage), 
and $\alpha:\mathbb{R}_+\to\mathbb{R}_+$ a non-negative and increasing function of time.
Then the volatility surface $(k,t)\mapsto w_\alpha(k,\theta_t):=w(k,\theta_t )+\alpha_t$ is also free of static arbitrage.
\end{theorem}


\begin{proof}
Since $\partial_t w_\alpha(k,\theta_t):=\partial_t w(k,\theta_t )+\partial_t \alpha_t$, 
Lemma~\ref{lem:noCalendarArb} implies that~$w_\alpha$ is free of calendar spread arbitrage if $\p_t\alpha_t \geq 0$ and $\alpha_t\geq 0$.
We now show that $w_\alpha$ is also free of butterfly arbitrage.
For clarity, since butterfly arbitrage does not depend on the time parameter~$t$, we shall use the simplified notation
$w(k):=w(k,\theta_t)$, and likewise $w_\alpha(k):=w_\alpha(k, \theta_t)$.
Similarly, in view of~\eqref{eq:gBut}, we shall define the map $g_\alpha(k)$, 
where the function~$w$ is replaced by~$w_\alpha$.
We consider the case $\rho<0$ since the case $\rho>0$ follows by symmetry, and the result is obvious when $\rho=0$.
Let us consider the function $G_\alpha:\mathbb{R}\to \mathbb{R}$ defined by
$$G_\alpha(k):=g(k)-g_\alpha(k), \quad\text{for all }k\in\mathbb{R},$$
and let $k^*:=-2\rho/\varphi(\theta_t)>0$ be the unique solution to the equation $w'(k)=0$.
We can compute explicitly the following:
$$
G_\alpha(k)=
\frac{w'(k)}{4}\left(\frac{1}{w_\alpha(k)}-\frac{1}{w(k)}\right)
\left(4k+w'(k)-w'(k)k^2\left(\frac{1}{w_\alpha(k)}+\frac{1}{w(k)}\right)\right),
$$
which implies

\begin{equation}\label{eq:DiffGalpha}
\partial_\alpha G_\alpha(k)= -\frac{w'(k)}{4}\frac{\left(w'(k)+4k\right)w_\alpha(k)-2k^2w'(k)}{w_\alpha(k)^3}.
\end{equation}
Since $w'(0)=\rho\theta_t\varphi(\theta_t)<0$ the equation $w'(k)+4k=0$ has a unique solution $k_*>0$,
and $w'(k)+4k$ is strictly positive for any $k>k_*$ and strictly negative when $k<k_*$.
By strict convexity of the function $w$ it also follows that $k_*<k^*$.
Therefore for any $k\in\left(k_*,k^*\right)$, the two inequalities
$w'(k)<0$ and $w'(k)+4k>0$
hold, and therefore $\partial_\alpha G_\alpha (k)>0$.
Since by construction $G_0(k)=0$, we therefore conclude that $g(k)>g_\alpha(k)$ for any $k\in\left(k_*,k^*\right)$.
For $k\notin\left(k_*,k^*\right)$, the inequality $g(k)<g_\alpha(k)$ holds as soon as $\partial_\alpha G_\alpha (k)<0$.
Consider first the case $k>k^*$.
We can rewrite~\eqref{eq:DiffGalpha} as
$$
\partial_\alpha G_\alpha(k)= -\frac{w'(k)}{4}\frac{2k\left[2w_\alpha(k)-kw'(k)\right]+w_\alpha(k)w'(k)}{w_\alpha(k)^3}.
$$
so that it suffices to prove the inequality $2w_\alpha(k)-kw'(k)>0$ for any $k>k^*$.
It suffices to prove $\partial_\alpha G_\alpha (k)<0$ for then we have the inequality $g_\alpha(k)>g(k) \geq 0$ and there is no butterfly arbitrage.

First consider the case $k>k^*$, so that $w'(k)>0$.  
Recall that a continuously differentiable function $f$ is convex on the interval $(a,b)$ if and only if
$f(x)-f(y)\geq f'(x)(x-y)$ for all $(x,y)\in (a,b)$.
Setting $x=k$ and $y=0$, we conclude that $2w_\alpha(k)-kw'(k)>0$ since $w_\alpha(0)\geq 0$.  It follows that $\partial_\alpha G_\alpha(k)<0$ for any $k>k^*$.

For any $k<0$, we always have $w'(k)<0$, the inequality $2w_\alpha(k)-kw'(k)>0$ follows by convexity as above, and hence
$\partial_\alpha G_\alpha(k)<0$ for any $k<0$.
We prove here that $g_\alpha(k)\geq g_\alpha(0)$ for all such $k$.
Since we already showed that $g_\alpha(0)>0$, the result follows.
From the definition of $g_\alpha$ and \eqref{eq:gBut},
\bea
g_\alpha(k)-g_\alpha(0)&=&\left(1-\frac{k w'(k)}{2 \,(w(k)+\alpha)}\right)^2-1\nonumber\\
&&-\frac{w'(k)^2}{4}\left(\frac{1}{w(k)+\alpha}+\frac{1}{4} \right) + \frac{w'(0)^2}{4}\left(\frac{1}{w(0)+\alpha}+\frac{1}{4} \right)\nonumber\\
&&+\frac{w''(k)}{2}-\frac{w''(0)}{2}.
\label{eq:galpha}
\eea
A straightforward analysis shows that the function $k\mapsto w''(k)$ is strictly increasing on the interval~$\left(0,k^*/2\right)$ and strictly decreasing on~$\left(k^*/2,k^*\right)$.  The easy computation $w''(0)=w''(k^*)$ implies that $w''(k)\geq w''(0)$ on $\left(0,k^*\right)$. Also, $w'(0)^2 > w'(k)^2$ on $(0,k^*)$.  Simplifying \eqref{eq:galpha}, it follows that
\beas
g_\alpha(k)-g_\alpha(0)&\geq&\left(1-\frac{k w'(k)}{2(w(k)+\alpha)}\right)^2-1
+ \frac{1}{4}\,\left(\frac{w'(0)^2}{w(0)+\alpha}-\frac{w'(k)^2}{w(k)+\alpha} \right)\nonumber \\
&\geq&\frac{1}{4}\,\left(\frac{w'(0)^2}{w(0)+\alpha}-\frac{w'(k)^2}{w(k)+\alpha} \right) 
-\frac{k\, w'(k)}{w(k)+\alpha}.
\label{eq:dg2}
\eeas
Note that $w'(k)^2 \leq w'(0)\,w'(k) \leq w'(0)^2$ on the interval $(0, k^*)$ so
\bea
g_\alpha(k)-g_\alpha(0)
&\geq&\frac{w'(0)\,w'(k)}{4}\,\left(\frac{1}{w(0)+\alpha}-\frac{1}{w(k)+\alpha} \right) 
-\frac{k\, w'(k)}{w(k)+\alpha}.
\label{eq:dg3}
\eea
We now prove the following claim:
$k w(0) - \frac{w'(0)}{4}[w(k)-w(0)] \geq 0$ for $k\in (0,k^*)$.
Indeed, 
$$
k w(0) - \frac{ w'(0)}{4}[w(k)-w(0)]
 = \left(1-\frac{\rho^2\,\theta\,\varphi^2}{8}\right)\theta \,k + \frac{\rho\, \varphi\,\theta^2}{8}
-\frac{\rho\, \varphi \,\theta^2}{8}\,\sqrt{\varphi^2 k^2 + 2\varphi\, \rho\, k +1}.
$$
Condition~2 of Theorem~\ref{thm:Butterfly} implies that $1-\frac{\rho^2\theta\varphi^2}{8}\geq 0$.
Then (recall that $\rho\leq 0$) the right-hand side of the above equality represents an increasing function on $(0,k^*)$ which is equal to zero at the origin, and the claim holds.
Then, from \eqref{eq:dg3},
\beas
g_\alpha(k)-g_\alpha(0)
&\geq&\frac{-w'(k)}{(w(0)+\alpha)\,(w(k)+\alpha)}\,\left\{k\, (w(0)+\alpha) - \frac{w'(0)}{4}\,\left[ w(k)-w(0) \right]
\right\}\\
&\geq& 0.
\eeas
\end{proof}

\begin{remark}
Given a set of expirations $0< t_1<\ldots<t_n$ ($n\geq 1$) and at-the-money implied total variances $0< \theta_{t_1}<\ldots<\theta_{t_n}$, Corollary \ref{cor:NoStaticArb} gives us the freedom to match three features of one smile (level, skew, and curvature say) but only two features of all the other smiles (level and skew say), subject of course to the given smiles being themselves arbitrage-free.  Theorem \ref{thm:AddingAlpha} may allow us to match an additional feature of each smile through $\alpha_t$.
\end{remark}

\section{Numerics and calibration methodology}\label{sec:calibration}
\subsection{How to eliminate butterfly arbitrage}

In Section~\ref{sec:sviSurfaces}, we showed how to define a volatility smile that is free of butterfly arbitrage.
This smile is completely defined given three observables.
The ATM volatility and ATM skew are obvious choices for two of them.
The most obvious choice for the third observable in equity markets would be the asymptotic slope for $k$ negative and in FX markets and interest rate markets, perhaps the ATM curvature of the smile might be more appropriate.

In view of Lemma~\ref{lem:ssviToJW}, supposing we choose to fix the SVI-JW parameters~$v_t$, $\psi_t$ and~$p_t$ of a given SVI smile, we may guarantee a smile with no butterfly arbitrage by choosing the  remaining parameters~$c_t'$ and~$\widetilde{v}'_{t}$ as
$$
c'_t = p_t+2\,\psi_t,
\quad\text{and}\quad
\widetilde{v}'_t = v_t\, \frac{4p_t c'_t}{\left(p_t+c'_t\right)^2}.
$$
In other words, given a smile defined in terms of its SVI-JW parameters, we are guaranteed to be able to eliminate butterfly arbitrage by changing the call wing~$c_t$ and the minimum variance~$\widetilde v_t$, both parameters that are hard to calibrate with available quotes in equity options markets.

\begin{example}\label{ex:VogtExample}
Consider again the arbitrageable smile from Example~\ref{ex:AxelVogt}.
The corresponding SVI-JW parameters read
$$
\left(v_t, \psi_t, p_t, c_t, \widetilde{v}_t\right) =
\left(0.01742625, -0.1752111, 0.6997381, 1.316798, 0.0116249\right).
$$
We know then that choosing
$\left(c_t, \widetilde v_t \right)=\left(c_t^{o}, \widetilde v_t^{o} \right):=\left(0.3493158, 0.01548182\right)$
gives a smile free of butterfly arbitrage.  
It follows by continuity of the parameterization in all of its parameters, that there must exist some pair of parameters $(c_t^*,\widetilde v_t^*)$ with $c_t^* \in (c_t^o, c_t)$ 
and $\widetilde v_t^* \in (\widetilde v_t,v_t^o)$ such that the new smile is free of butterfly arbitrage and is as close as possible to the original one in some sense.  In this particular case, choosing the objective function as the sum of squared option price differences plus a large penalty for butterfly arbitrage, we arrive at the following ``optimal''  choices of the call wing and minimum variance parameters that still ensure no butterfly arbitrage:
$$\left(c_t^*, \widetilde v_t^*\right) = \left(0.8564763, 0.0116249\right).$$
Note that the optimizer has left~$\widetilde v_t$ unchanged but has decreased the call wing.
The resulting smiles and plots of the function $g$ are shown in Figure~\ref{fig:VogtVolFixed}.

\begin{figure}[htb!]
\begin{center}
\subfigure{\includegraphics[scale=0.27]{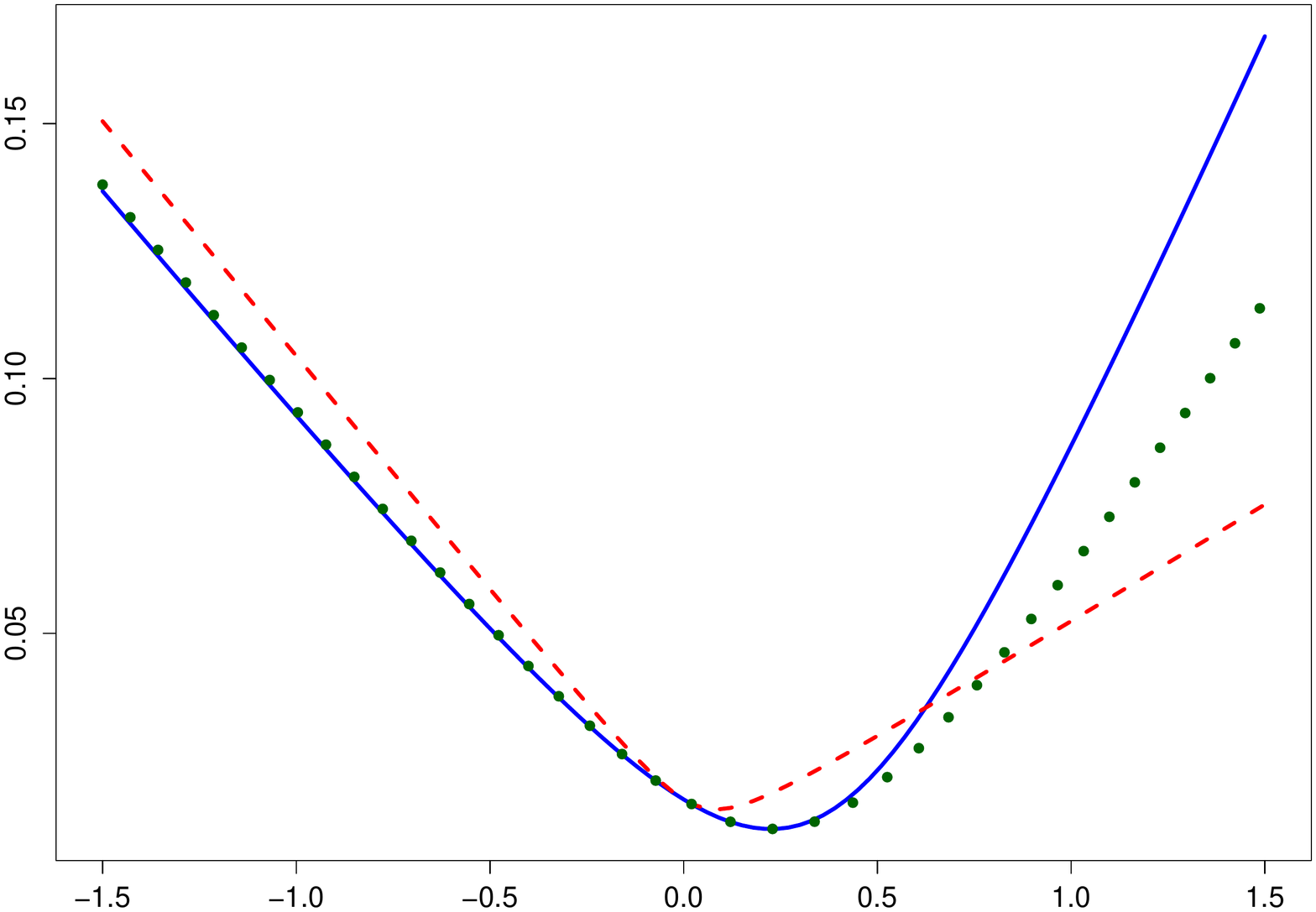}}
\subfigure{\includegraphics[scale=0.27]{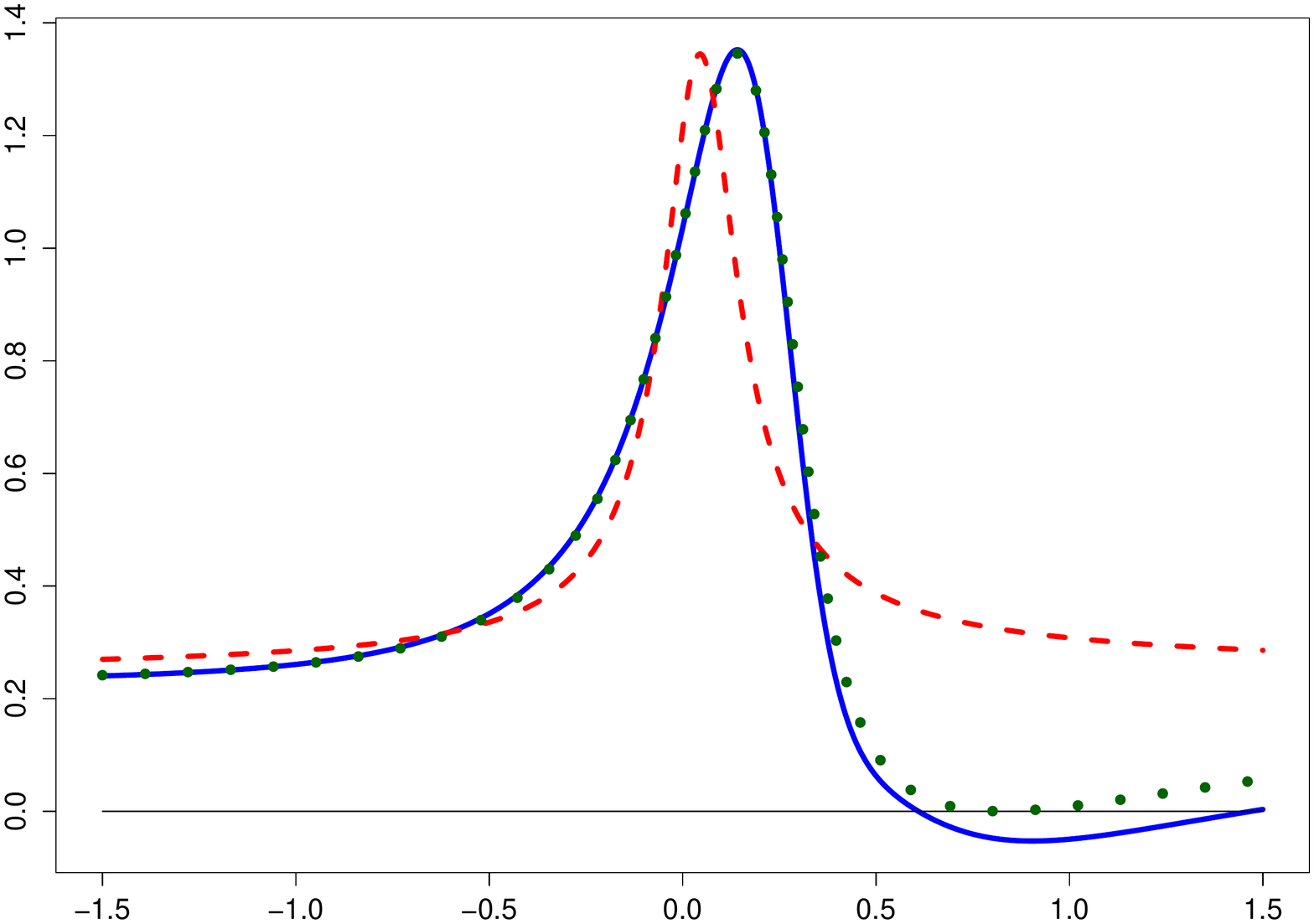}}
\caption{Plots of the total variance smile (left) and the function $g$ defined in~\eqref{eq:gBut} (right), using the parameters~\eqref{eq:VogtParams}.
The graphs corresponding to the original Vogt parameters is solid, to the guaranteed butterfly-arbitrage-free parameters dashed, and to the ``optimal'' choice of parameters dotted.}
\label{fig:VogtVolFixed}
\end{center}
\end{figure}
\end{example}

\begin{remark}
The additional flexibility potentially afforded to us through the parameter $\alpha_t$ of Theorem \ref{thm:AddingAlpha} sadly does not help us with the Vogt smile of Example \ref{ex:VogtExample}.  For $\alpha_t$ to help, we must have $\alpha_t>0$; it is straightforward to verify that this translates to the condition $v_t\,(1-\rho^2) < \tilde v_t$ which is violated in the Vogt case.
\end{remark}

\subsection{Calibration of SVI parameters to implied volatility data}\label{sec:fitting}

There are many possible ways of defining an objective function, the minimization of which would permit us to calibrate SVI to observed implied volatilities.  Whichever calibration strategy we choose, we need an efficient fitting algorithm and a good choice of initial guess.
The approach we will present here involves taking a square-root fit as the initial guess.  We then fit SVI slice-by-slice with a heavy penalty for calendar spread arbitrage (i.e. crossed lines on a total variance plot).
Consider two SVI slices with parameters $\chi_1$ and $\chi_2$ where $t_2>t_1$.  
We first compute the points $k_i$ $(i=1,\ldots,n)$ with $n\leq 4$ at which the slices cross, sorting them in increasing order.  If $n>0$, we define the points $\widetilde k_i$ as
\beas
\widetilde k_1&:=&k_1-1,\\
\widetilde k_i&:=&\frac 12\,(k_{i-1}+k_{i}),\quad \text{if } 2 \leq i \leq n, \\
\widetilde k_{n+1}&:=&k_n+1.
\eeas
For each of the $n+1$ points $\widetilde k_i$, we compute the amounts $c_i$ by which the slices cross:
\[
c_i=\max\left[0,w(\widetilde k_i;\chi_1)-w(\widetilde k_i;\chi_2)\right].
\]

\begin{definition}\label{def:crossedness}
The {\em crossedness} of two SVI slices is defined as the maximum of the $c_i$ $ (i=1,\ldots,n)$.  
If $n=0$, the crossedness is null.
\end{definition}

\begin{center}
\colorbox{shade}{\begin{minipage}[h!]{0.95\textwidth}

\begin{center}
{\bf An example SVI calibration recipe}
\end{center}

\begin{itemize}
\item Given mid implied volatilities $\sigma_{ij}=\sigma_{\BS}(k_i,t_j)$, compute mid option prices using the Black-Scholes formula.
\item Fit the square-root SVI surface by minimizing sum of squared distances between the fitted prices and the mid option prices.   This is now the initial guess.
\item Starting with the square-root SVI initial guess, change SVI parameters slice-by slice so as to minimize the  sum of squared distances between the fitted prices and the mid option prices with a big penalty for crossing either the previous slice or the next slice (as quantified by the crossedness from Definition~\ref{def:crossedness}).
\end{itemize}
\end{minipage}}
\end{center}

There are obviously many possible variations on this recipe.  The objective function may be changed and when finally working to optimize the fit slice-by-slice, one can work from the shortest expiration to the longest expiration or in the reverse order.  In practice, working forward or in reverse seems to make little difference. Changing the objective function on the other hand will make some difference especially for very short expirations.

\subsection{Interpolation and extrapolation of calibrated slices}

{Suppose we follow the above recipe above to fit SVI to options with a discrete set of expiries.  In particular, each of the resulting SVI smiles will be free of butterfly arbitrage.  It's not immediately obvious that we can interpolate these smiles in such a way as to ensure the absence of static arbitrage in the interpolated surface.  The following lemma shows that it is possible to achieve this.}

\begin{lemma}\label{lem:interpolation}
Given two volatility smiles~$w(k,t_1)$ and~$w(k,t_2)$ with $t_1 < t_2$ where the two smiles are free of butterfly arbitrage and such that
$w(k,\tau_2)\geq w(k,\tau_1)$ for all~$k$, there exists an interpolation such that the interpolated volatility surface is free of static arbitrage for $t_1<t<t_2$.
\end{lemma}

\begin{proof}
Given the two smiles $w(k,t_1)$ and $w(k,t_2)$, we may compute the (undiscounted) prices $C(F_i,K_i,t_i)=:C_i$ of European calls with expirations $t_i$ ($i=1,2$) using the Black-Scholes formula.
In particular, since the two smiles are free of butterfly arbitrage,
\[
\frac{\p ^2 C_i}{\p K^2} \geq 0,\quad \text{ for }i=1,2.
\]
Consider any monotonic interpolation $\theta_t$ of the at-the-money implied total variance $w(0,t)$.
Let $K_i=F_i\Ex^k$ and $K_t=F_t\Ex^k$.
Then for any $t_1<t<t_2$, define the price $C_t=C(F_t,K_t,t)$ of a European call option to be
\beq
\frac{C_t}{K_t}=\alpha_t\frac{C_1}{K_1}+\left(1-\alpha_t\right)\frac{C_2}{K_2},
\label{eq:interpolatedCall}
\eeq
where for any $t\in \left(t_1,t_2\right)$, we define
\beq
\alpha_t:=\frac{\sqrt{\theta_{t_2}}-\sqrt{\theta_{t}}}{\sqrt{\theta_{t_2}}-\sqrt{\theta_{t_1}}}\in\left[0,1\right].
\label{eq:alphat}
\eeq
By construction, for fixed $k$, the inequality
\[
\frac{\p}{\p\tau} \frac{C_t}{K_t} \geq 0
\]
holds so that there is no calendar spread arbitrage.  Also, because of the square-roots in the definition \eqref{eq:alphat}, the at-the-money interpolated option price will be almost perfectly consistent with the chosen implied total variance interpolation $\theta_{t}$.
Moreover, if the two smiles $w(k,t_1)$ and $w(k,t_2)$ are free of butterfly arbitrage, we have $\p_{K,K} C(k,t) \geq 0$.
To see this, first note that because all the options have the same moneyness, the identity~\eqref{eq:interpolatedCall} is equivalent to
\beq
\frac{C_t}{F_t}=\alpha_t\frac{C_1}{F_1}+\left(1-\alpha_t\right)\frac{C_2}{F_2}.
\label{eq:interpolated2}
\eeq
Then note that the ratio $C(F,K,t)/F$ is a function of $F$ and $K$ only through the log-moneyness~$k$.
Also, for $K=K_t, K_1, K_2$, we have
\[
K^2\frac{\p^2 f}{\p K^2}=\frac{\p^2 f}{\p k^2}-\frac{\p f }{\p k}.
\]
Applying this to~\eqref{eq:interpolated2}, we obtain
$$
\frac{K_\tau^2}{F_t}\frac{\partial^2 C_t}{\partial K_t^2}
=\alpha_t\frac{K_1^2}{F_1}\frac{\partial^2 C_1}{\partial K_1^2}
+\left(1-\alpha_t\right)\frac{K_2^2}{F_2}\frac{\partial^2 C_2}{\partial K^2}.
$$
All the terms on the rhs are non-negative, so the lhs must also be non-negative. We conclude that there is no butterfly arbitrage in the interpolated slice and thus that there is no static arbitrage.
The interpolated volatility surface may be retrieved by inversion of the Black-Scholes formula.
\end{proof}

We could conceive of a myriad of algorithms for extrapolating the volatility surface.
For example, one way to extrapolate a given set of $n\geq 1$ (arbitrage-free) volatility smiles
with expirations $0<t_1<\ldots<t_n$ would be as follows:
At time $t_0=0$, the value of a call option is just the intrinsic value.  
We may then interpolate between $t_0$ and $t_1$ using the algorithm presented in the proof of Lemma~\ref{lem:interpolation}, thereby guaranteeing no static arbitrage.  
For extrapolation beyond the final slice, we suggest to first recalibrate the final slice 
using the SSVI form~\eqref{eq:ssvi}. 
Then fix a monotonic increasing extrapolation of~$\theta_t$
(asymptotically linear in time would seem to be reasonable) and extrapolate the smile for $t>t_n$ according to
\[
w(k,\theta_t)=w(k,\theta_{t_n})+\theta_t-\theta_{t_n},
\]
which is free of static arbitrage if $w(k,\theta_{t_n})$ is free of butterfly arbitrage by Theorem~\ref{thm:AddingAlpha}.

\subsection{A calibration example}

We take SPX option quotes as of 3pm on 15-Sep-2011 (the day before triple-witching) and compute implied volatilities for all 14 expirations. The result of fitting square-root SVI is shown in Figure \ref{fig:sqrtSVI}.  The result of fitting SVI following the recipe provided in Section \ref{sec:fitting} is shown in Figure \ref{fig:SVIQR}.  With the sole exception of the first expiration (options expiring at the market open on the following morning), the fit quality is almost perfect.

\begin{figure}[htb!]
\begin{center}
\includegraphics[width=\linewidth]{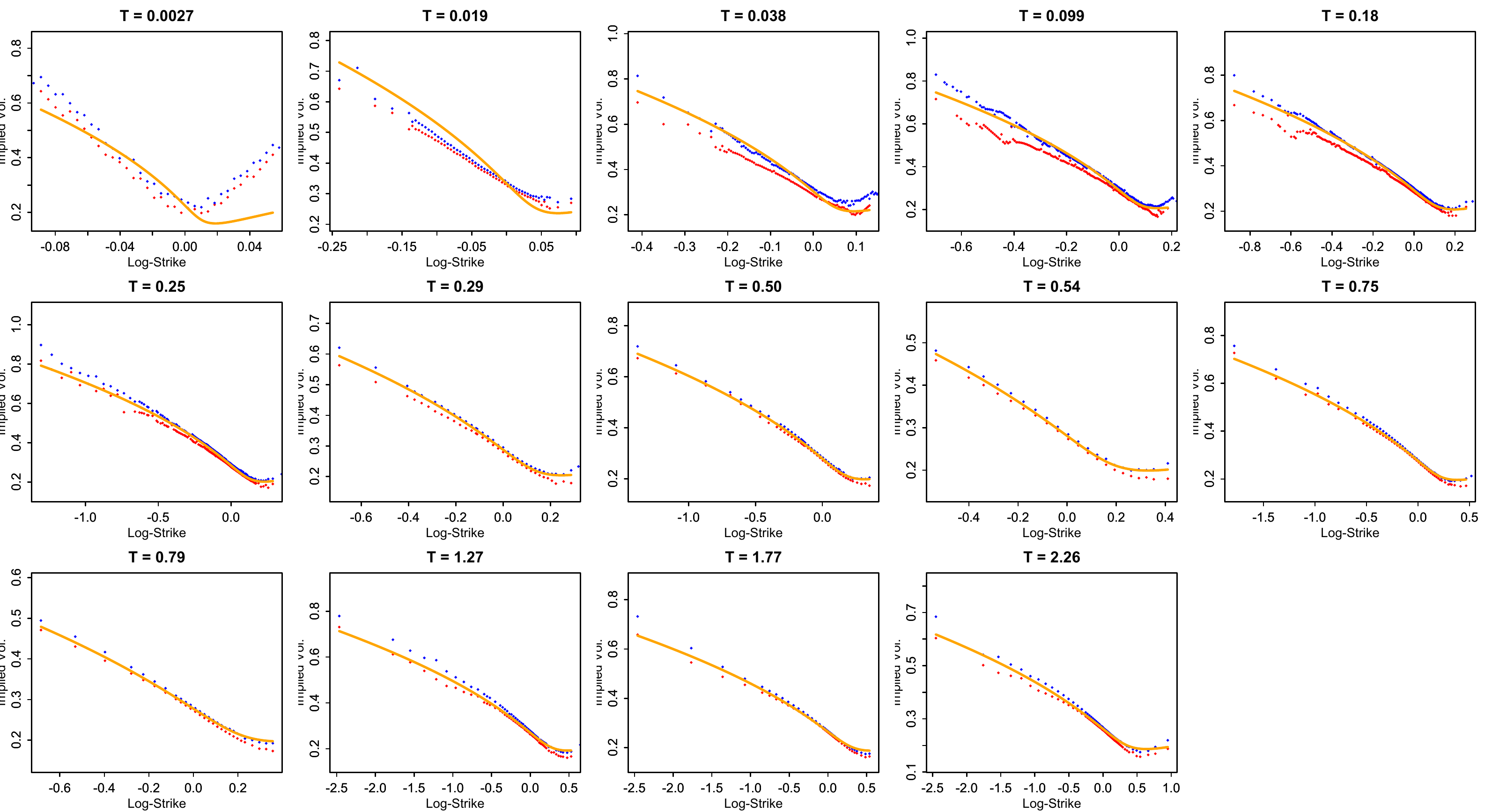}
\caption{Red dots are bid implied volatilities; blue dots are offered implied volatilities; the orange solid line is the square-root SVI fit}
\label{fig:sqrtSVI}
\end{center}
\end{figure}

\begin{figure}[htb!]
\begin{center}
\includegraphics[width=\linewidth]{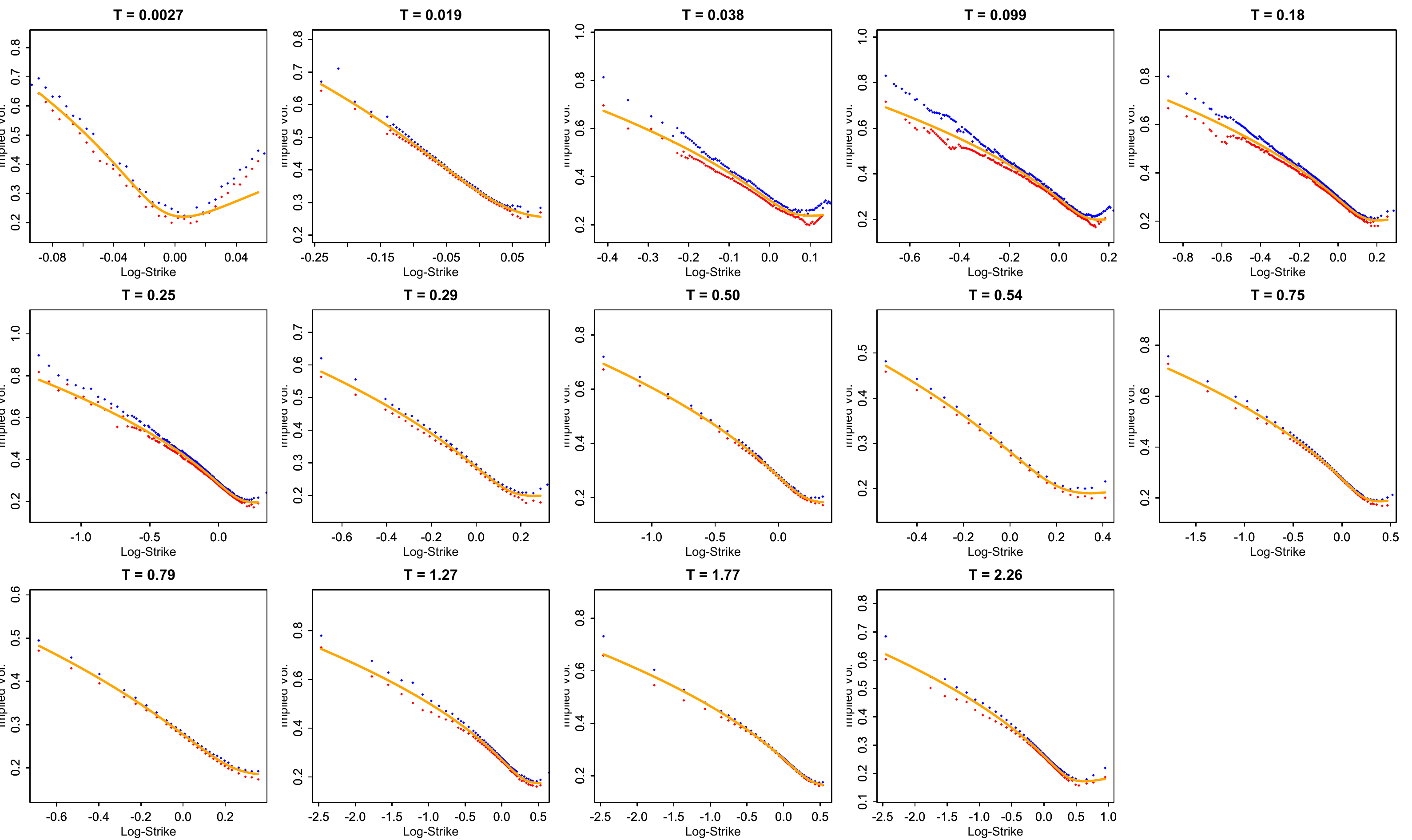}
\caption{Red dots are bid implied volatilities; blue dots are offered implied volatilities; the orange solid line is the SVI fit following recipe of Section \ref{sec:fitting}}
\label{fig:SVIQR}
\end{center}
\end{figure}

\newpage

\section{Summary and conclusion}\label{sec:conclusion}

We have found and described a large class of arbitrage-free SVI volatility surfaces with a simple closed-form representation.  Taking advantage of the existence of such surfaces, we showed how to eliminate both calendar spread and butterfly arbitrages when calibrating SVI to implied volatility data.  We have also demonstrated the high quality of typical SVI fits with a numerical example using recent SPX options data.
The potential applications of this work to modelling the dynamics of the implied volatility surface are left for future research.

\section*{Acknowledgments}

The first author is very grateful to his former colleagues at Bank of America Merrill Lynch for their work on SVI and its implementation, in particular Chrif Youssfi and Peter Friz.  
We also thank Richard Holowczak of the Subotnick Financial Services Center at Baruch College for supplying the SPX options data,
Andrew Chang of the Baruch MFE program for helping with the data analysis,
Julien Guyon and the participants of Global Derivatives, Barcelona 2012 for their feedback and comments.
We are very grateful to the anonymous referees for their helpful comments and suggestions, and in particular to one of the referees who led us to tighten our results and correct an error in one proof.

%
%
%
%

\end{document}